\documentclass[10pt,conference,english,twocolumn]{IEEEtran} %

\usepackage{babel}
\usepackage[babel]{microtype}
\usepackage[babel]{csquotes}
\usepackage[svgnames]{xcolor}
\usepackage{graphicx}
\usepackage{tikz}
\usepackage{pgfplots}

\usepackage{tabularx}
\usepackage{booktabs}

\usepackage{stfloats}

\usepackage{amsmath}
\usepackage{amsfonts}
\usepackage{amssymb}
\usepackage{amsthm}
\usepackage{bm}
\usepackage{siunitx}
\usepackage{mathtools}
\usepackage{dsfont}

\usepackage{algorithm}
\usepackage{algpseudocode}

\usepackage[math]{blindtext}

\usepackage[backend=biber,url=false,style=ieee,dashed=false,isbn=false,doi=false,eprint=true]{biblatex}
\bibliography{literature.bib}
\renewcommand*{\bibfont}{\small}

\usepackage{hyperref}
\hypersetup{
	pdfauthor={},
	pdftitle={},
	bookmarksnumbered=true,
	colorlinks=true,
	allcolors=.,
	urlcolor=MidnightBlue,
}

\usepackage[acronyms,nomain,xindy,nowarn]{glossaries}
\makeglossaries
\newacronym{5g}{5G}{Fifth-Generation Wireless Systems}
\newacronym{ask}{ASK}{amplitude-shift keying}
\newacronym[plural={AWGN}, \glsshortpluralkey={AWGN}]{awgn}{AWGN}{additive white Gaussian noise}

\newacronym{ber}{BER}{bit error rate}
\newacronym{bler}{BLER}{block error rate}
\newacronym{bpsk}{BPSK}{binary phase-shift keying}

\newacronym{cdf}{CDF}{cumulative distribution function}
\newacronym{cnn}{CNN}{convolutional neural network}
\newacronym{csi}{CSI}{channel state information}
\newacronym{csit}{CSI\babelhyphen{nobreak}T}{channel-state information at the transmitter}
\newacronym{csir}{CSI\babelhyphen{nobreak}R}{channel-state information at the receiver}

\newacronym{ecc}{ECC}{error-correcting code}
\newacronym{ecdf}{ECDF}{empirical distribution function}
\newacronym{ee}{EE}{energy efficiency}

\newacronym{fdd}{FDD}{frequency-division duplex}
\newacronym{ff}{FF}{feed-forward}
\newacronym{ffnn}{FF-NN}{feed-forward neural network}
\newacronym{fsk}{FSK}{frequency-shift keying}

\newacronym{gm}{GM}{Gaussian mixture}

\newacronym{iid}{i.i.d.}{independent and identically distributed}
\newacronym{il}{IL}{information leakage}
\newacronym{iot}{IoT}{internet of things}

\newacronym{lb}{LB}{lower bound}
\newacronym{los}{LoS}{line-of-sight}

\newacronym{mac}{MAC}{multiple access channel}
\newacronym{map}{MAP}{maximum a posteriori}
\newacronym{mc}{MC}{Monte Carlo}
\newacronym{mimo}{MIMO}{multiple-input multiple-output}
\newacronym{mkp}{MKP}{multiple knapsack problem}
\newacronym{ml}{ML}{machine learning}
\newacronym{mmwave}{mmWave}{millimeter wave}
\newacronym{mop}{MOP}{multi-objective programming problem}
\newacronym{mrc}{MRC}{maximum ratio combining}
\newacronym{msc}{MSC}{modulation and coding scheme}
\newacronym{mse}{MSE}{mean squared error}

\newacronym{nlos}{NLoS}{non-line-of-sight}
\newacronym{nn}{NN}{neural network}
\newacronym{nve}{NVE}{normalized validation error}

\newacronym{pdf}{PDF}{probability density function}
\newacronym{pwc}{PWC}{polar wiretap code}

\newacronym{qam}{QAM}{quadrature amplitude modulation}
\newacronym{qmkp}{QMKP}{quadratic multiple knapsack problem}
\newacronym{qmkp-hp}{QMKP-HP}{quadratic multiple knapsack problem with heterogeneous profits}

\newacronym{ra}{RA}{rearrangement algorithm}
\newacronym{rbf}{RBF}{radial basis function}
\newacronym{relu}{ReLU}{rectified linear unit}
\newacronym{ris}{RIS}{reconfigurable intelligent surface}
\newacronym{rnn}{RNN}{recurrent neural network}
\newacronym{rr}{RR}{round-robin}

\newacronym{sde}{SDE}{stochastic differential equation}
\newacronym{sgd}{SGD}{stochastic gradient descent}
\newacronym{sgp}{SGP}{secure goodput}
\newacronym{simo}{SIMO}{single-input multiple-output}
\newacronym{sinr}{SINR}{signal-to-interference-plus-noise ratio}
\newacronym{siso}{SISO}{single-input single-output}
\newacronym{snr}{SNR}{signal-to-noise ratio}
\newacronym{svm}{SVM}{support vector machine}

\newacronym{tvd}{TVD}{total variation distance}

\newacronym{uav}{UAV}{unmanned aerial vehicle}
\newacronym{ub}{UB}{upper bound}
\newacronym{urllc}{URLLC}{ultra-reliable low latency communication}

\newacronym{v2v}{V2V}{vehicle-to-vehicle}

\newacronym{zoc}{ZOC}{zero-outage capacity}
\setacronymstyle{long-short}
\glsdisablehyper

\definecolor{plot0}{HTML}{004488}
\definecolor{plot1}{HTML}{DDAA33}
\definecolor{plot2}{HTML}{BB5566}
\definecolor{plot3}{HTML}{000000}
\definecolor{plot4}{HTML}{AAAAAA}

\pgfplotscreateplotcyclelist{colorcycle}{ %
	{plot0, mark=*, thick, mark options=solid},
	{plot1, mark=triangle*, thick, mark options=solid},
	{plot2, mark=square*, thick, mark options=solid},
	{plot3, mark=diamond*, thick, mark options=solid},
	{plot4, mark=pentagon*, thick, mark options=solid},
}
\DeclarePairedDelimiter{\abs}{\vert}{\vert}

\DeclarePairedDelimiter{\floor}{\lfloor}{\rfloor}

\newcommand*{\imag}{\ensuremath{\mathrm{j}}}
\newcommand*{\e}{\ensuremath{\mathrm{e}}}

\newcommand*{\naturals}{\ensuremath{\mathds{N}}}

\newcommand*{\unif}{\ensuremath{\mathrm{unif}}}

\let\set\mathcal
\DeclareMathOperator{\vd}{vd}

\let\setintersect\cap

\DeclareMathOperator{\hilbert}{\mathcal{H}}
\pgfplotsset{compat=newest}
\usetikzlibrary{positioning,calc,external}

\DeclareSIUnit{\dBm}{dBm}
\newcommand{\todo}[2][]{
	\if\relax\detokenize{#1}\relax
	{\color{red}[TODO: #2]}%
	\else
	{\color{red}[TODO (#1): #2]}%
	\fi
}

\theoremstyle{plain}%
\newtheorem{thm}{Theorem}

\newtheorem{lem}{Lemma}

\theoremstyle{definition}
\newtheorem{prob}{Problem Statement}
\newtheorem*{prob*}{Problem Statement}

\theoremstyle{remark}
\newtheorem{rem}{Remark}
\newtheorem*{rem*}{Remark}

\newtheoremstyle{example}{\topsep}{\topsep}{}{}{\itshape}{.}{ }{}
\theoremstyle{example}
\newtheorem{example}{Example}
\newtheorem*{example*}{Example}

\algnewcommand\algorithmictry{\textbf{try}}
\algnewcommand\algorithmicendtry{\textbf{end try}}
\algnewcommand\algorithmicexcept{\textbf{except}}
\algblock{Try}{EndTry}
\algcblock[Try]{Try}{Except}{EndIf}
\algrenewtext{Try}{\algorithmictry}
\algrenewtext{EndTry}{\algorithmicendtry}
\algrenewtext{Except}{\algorithmicexcept}

\algnewcommand{\algorithmiclongcomment}[1]{\textcolor{gray}{#1}}
\algrenewcommand{\algorithmiccomment}[1]{\hfill\algorithmiclongcomment{//~#1}}
\algnewcommand{\LongComment}[1]{\algorithmiclongcomment{/*~#1~*/}}

\algrenewcommand{\algorithmicindent}{1em}
\newcommand*{\len}[1][]{%
	\if\relax\detokenize{#1}\relax
	\ensuremath{\ell}
	\else
	\ensuremath{\ell_{#1}}
	\fi
}
\newcommand*{\lenref}[1][]{%
	\if\relax\detokenize{#1}\relax
	\ensuremath{\tilde{\ell}}
	\else
	\ensuremath{\tilde{\ell}_{#1}}
	\fi
}

\newcommand*{\htx}[1][]{%
	\if\relax\detokenize{#1}\relax
	\ensuremath{h_{\text{Tx}}}
	\else
	\ensuremath{h_{\text{Tx}, #1}}
	\fi
}
\newcommand*{\hrx}[1][]{%
	\if\relax\detokenize{#1}\relax
	\ensuremath{h_{\text{Rx}}}
	\else
	\ensuremath{h_{\text{Rx}, #1}}
	\fi
}
\newcommand*{\dmin}[1][]{%
	\if\relax\detokenize{#1}\relax
	\ensuremath{d_{\text{min}}}
	\else
	\ensuremath{d_{\text{min},#1}}
	\fi
}
\newcommand*{\dmax}[1][]{%
	\if\relax\detokenize{#1}\relax
	\ensuremath{d_{\text{max}}}
	\else
	\ensuremath{d_{\text{max},#1}}
	\fi
}
\newcommand*{\domega}{\ensuremath{\Delta\omega}}
\let\dw\domega

\newcommand*{\df}{\ensuremath{\Delta f}}

\newcommand*{\dphi}{\ensuremath{\Delta\phi}}
\let\w\omega

\newcommand*{\recpower}[1][]{%
	\if\relax\detokenize{#1}\relax
	\ensuremath{P_{r}}
	\else
	\ensuremath{P_{r,#1}}
	\fi
}
\newcommand*{\sumpower}[1][]{%
	\if\relax\detokenize{#1}\relax
	\ensuremath{P_{s}}
	\else
	\ensuremath{P_{s,#1}}
	\fi
}

\newcommand*{\users}{\ensuremath{\set{U}}}
\newcommand*{\frequencies}{\ensuremath{\set{F}}}
\newcommand*{\profits}{\ensuremath{\set{P}}}
\newcommand*{\assignments}{\ensuremath{\set{A}}}

\title{Multi-User Frequency Assignment for\\Ultra-Reliable mmWave Two-Ray Channels}
\author{%
	\IEEEauthorblockN{Karl-Ludwig Besser and Eduard A. Jorswieck}
	\IEEEauthorblockA{Institute for Communications Technology\\Technische Universit\"at Braunschweig, Germany\\Email: \{k.besser, e.jorswieck\}@tu-bs.de}
	\and
	\IEEEauthorblockN{Justin P. Coon}
	\IEEEauthorblockA{Department of Engineering Science\\University of Oxford, U.\,K.\\ Email: justin.coon@eng.ox.ac.uk}
}

\makeatletter
\def\thanksfootnote{\gdef\@thefnmark{}\@footnotetext}
\makeatother

\begin{document}
\maketitle

\begin{abstract}\noindent
	We consider a multi-user two-ray ground reflection scenario with unknown distances between transmitter and receivers.
	By using two frequencies per user in parallel, we can mitigate possible destructive interference and ensure ultra-reliability with only very limited knowledge at the transmitter.
	In this work, we consider the problem of assigning two frequencies to each receiver in a multi-user communication system such that the average minimum receive power is maximized.
	In order to solve this problem, we introduce a generalization of the quadratic multiple knapsack problem to include heterogeneous profits and develop an algorithm to solve it.
	Compared to random frequency assignment, we report a gain of around 6\,dB in numerical simulations.
\end{abstract}
\begin{IEEEkeywords}
	Ultra-reliable communications, Two-ray ground reflection, Knapsack problems, mmWave, Worst-case design.
\end{IEEEkeywords}
\thanksfootnote{The work of E. Jorswieck is partly supported by the Federal	Ministry of Education and Research Germany (BMBF) as part of the 6G Research and Innovation Cluster 6G-RIC under Grant 16KISK020K. The work of J. Coon is supported by the EPSRC under grant number EP/T02612X/1.}
\glsresetall

\section{Introduction}\label{sec:introduction}
Reliability is a major requirement for many modern applications of wireless communication systems~\cite{Saad2020,Park2022}.
In particular, this includes autonomous vehicles, e.g., self-driving cars and \glspl{uav}.
It is therefore of great interest to develop techniques, which enable ultra-reliable communications.
This is especially important for scenarios where only limited information, e.g., \gls{csi}, is available at the communication parties, e.g., due to high mobility or in \gls{fdd} systems.

Among others, it has been observed that negative dependency between channel gains can significantly improve reliability~\cite{Haber1974,Besser2020twc,Besser2021zoc}.
The basic idea is to establish diversity and ensure that always one communication link is available, if the others fail.
In the following, we will use this idea to develop a simple frequency diversity scheme that enables ultra-reliable communications in two-ray ground reflection scenarios.
In this two-ray model, it is assumed that only one significant multipath component exists in addition to a \gls{los} connection.
The second component is typically caused by a single reflection on a ground surface.
This could occur in flat outdoor terrain~\cite{Weiler2015}, on large concrete areas, e.g., airports~\cite{Naganawa2017}, and for a \gls{uav} flying above water~\cite{Matolak2017,Chiu2021}.
It has also been observed that the two-ray model can be appropriate for \gls{v2v} communication scenarios~\cite{Sommer2012,Farzamiyan2020}.
This includes high frequency bands like \gls{mmwave}~\cite{Guan2017,Khawaja2020,Zochmann2017}.

In general, the curvature of the Earth's surface needs to be considered for long-distance outdoor settings and accurate models like the curved-Earth model~\cite{Matolak2015,Matolak2017,Parsons2001} exist.
However, when considering relatively short distances, the flat-Earth model is a valid approximation~\cite{Matolak2017,Parsons2001}, which we adopt throughout this work.

When varying the distance between transmitter and receiver, the relative phase of the two received signal components varies and they may interfere constructively or destructively.
A destructive interference causes a drop of receive power, which in turn could cause an outage of the communication link.
In order to mitigate drops of the signal power on one frequency, a second frequency can be used in parallel.
The use of multiple frequencies in parallel to create diversity and improve the reliability in ground reflection scenarios has already been proposed in \cite{Haber1974} and \cite{Berger1972diversity}.
In \cite{Capriglione2015}, it is analyzed and experimentally verified that frequency diversity improves the performance of distance measurements in outdoor ground reflection scenarios.
Instead of using multiple frequencies in parallel, it was shown experimentally in \cite{Naganawa2017} that using multiple antennas and carefully choosing the spacing between them can also improve the received power.

In practical communication systems, we are often given a fixed set of available frequencies.
If there exist multiple users with unknown distances, the following questions immediately arises: \emph{how should the frequencies be assigned to the users in order to maximize the worst-case receive power?}
Such an optimization problem of assigning discrete items to users belongs to the broad class of knapsack problems~\cite{Kellerer2004}.

In this work, we consider the problem of assigning each user in a multi-user communication system two frequencies such that the average worst-case receive power is maximized.
Our main contributions are summarized as follows.
\begin{itemize}
	\item We analyze the worst-case received power for a two-ray ground reflection model with unknown distance between transmitter and receiver when employing only a single frequency (Section~\ref{sub:single-freq}) and when employing two frequencies in parallel (Section~\ref{sub:two-freq}).
	\item We formulate the problem of assigning frequencies to users with unknown distances as a knapsack problem. (Section~\ref{sec:freq-assign-qmkp})
	\item \vspace*{-.08in}In order to solve this problem, we introduce a novel generalization of the standard \gls{qmkp} to support heterogeneous profits. We then present an algorithm to find a solution to it. (Section~\ref{sec:qmkp-hp})
	\item Finally, we demonstrate the effectiveness of our proposed solution by numerical examples and report a gain of up to \SI{6}{\decibel} compared to a random frequency assignment. (Section~\ref{sec:numerical-example})
\end{itemize}

In \cite{Besser2022tworay}, we analyze a similar two-ray model with two frequencies in parallel.
However, we only consider a single user and focus on optimizing the frequency spacing instead of assuming a fixed set of available frequencies.

\section*{Notation}
An overview of the most commonly used variable notation can be found in Table~\ref{tab:notation}.

\vspace*{-.05in}
\begin{table}[!ht]
	\renewcommand*{\arraystretch}{1.2}
	\caption{Notation of the Most Commonly Used Variables and System Parameters}\vspace*{-.08in}
	\label{tab:notation}
	\begin{tabularx}{\linewidth}{lX}
		\toprule
		$\users=\{1, 2, \dots{}, K\}$ & Set of users\\
		$\frequencies=\{f_1, f_2, \dots{}, f_N\}$ & Set of available frequencies\\
		$\profits=\{P_1, P_2, \dots{}, P_K\}$ & Profit $P_u$ for user $u$\\
		$\assignments=\{\assignments_1, \dots{}, \assignments_K\}$ & Set $\assignments_u$ of assigned frequencies to user $u$\\%[\smallskipamount]
		$d_u$ & Distance between transmitter and receiver~$u$ (on the ground) [$\si{\meter}$]\\
		$\htx$ & Height of the transmitter [$\si{\meter}$]\\
		$\hrx[u]$ & Height of receiver~$u$ [$\si{\meter}$]\\
		$\len[u]$ & Length of the \gls{los} path to user~$u$ [$\si{\meter}$]\\
		$\lenref[u]$ & Total length of reflection path to user~$u$ [$\si{\meter}$]\\%[\smallskipamount]
		$c$ & Speed of light $=\SI{299792458}{\meter\per\second}$\\
		$\omega = 2\pi f$ & (Angular) frequency ([$\si{\radian\per\second}$]) [$\si{\hertz}$]\\
		$P_{t}$ & Transmit power [$\si{\watt}$]\\
		$P_{r}$ & Receive power (single frequency) [$\si{\watt}$]\\
		$\sumpower$ & Received sum power (two frequencies) [$\si{\watt}$]\\
		$\underline{P_{s}}$ & Lower bound of the receive sum power [$\si{\watt}$]
		\\%[\smallskipamount]
		$d_k$ & Distance at which the $k$-th local minimum of the receive power occurs [\si{\meter}]\\
		$\dw=2\pi\df=\omega_2-\omega_1$ & (Angular) frequency spacing ([$\si{\radian\per\second}$]) [$\si{\hertz}$]\\
		\bottomrule
	\end{tabularx}
\end{table}

In order to simplify the notation, we will omit variables on which functions depend when their value is clear from the context, e.g., we will write $f(x)$ instead of $f(x, y)$ when the value of $y$ is fixed.

Since the angular frequency $\omega=2\pi f$ is a simple scaling of the frequency $f$, we will treat them somewhat interchangeably.
Especially for calculations, it is more convenient to use $\w$, while $f$ is relevant for actual system design.
We will therefore use the frequency $f$ for the numerical examples while expressing all formulas in terms of the angular frequency $\w$.

The uniform distribution on the interval $[a,b]$ is denoted as $\unif[a,b]$.

\section{System Model and Problem Formulation}\label{sec:system-model}
We consider a multi-user communication system with $K$ single-antenna users at unknown distances $d_u$ from the single-antenna transmitter.
For each user $u$, only a range of possible distances is known, i.e., it is known that $d_u\in[\dmin[u], \dmax[u]]$.
An illustration can be found in Fig.~\ref{fig:illustration-multi-user-distances}.
It should be noted that the intervals of the individual users may overlap or even be the same.
\begin{figure}
	\centering
	\begin{tikzpicture}
	\node[draw,circle,fill=black, label=left:{Transmitter}] (tx) at (0, 0) {};
	
	\node[draw, circle, label=left:{Receiver $u$}] (rx1) at (50:3) {};
	\draw[|<->|] (tx) -- node[below right] {$d_u$} (rx1);
	
	\draw[plot0, dashed,thick] (0:1) arc (0:170:1);
	\draw[plot1, dashed,thick] (0:4) arc (0:180:4);
	
	\draw[|<->|,plot0] (tx) -- node[below] {$\dmin[u]$} (0:1);
	\draw[|<->|,plot1] (tx) -- node[left] {$\dmax[u]$} (120:4);

\end{tikzpicture}
	\vspace*{-1em}
	\caption{Illustration of the considered scenario. Each user $u$ is at an unknown distance~$d_u$ from the transmitter within the known interval $[\dmin[u], \dmax[u]]$.}
	\label{fig:illustration-multi-user-distances}
\end{figure}
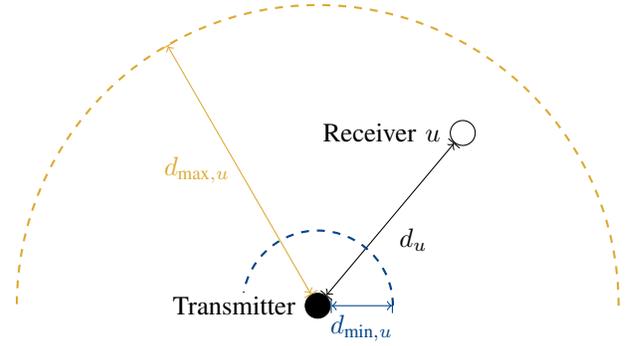
The transmitter is fixed at height $\htx$ above the ground while the users are located at heights $\hrx[u]$ above the ground.

The scenario is considered to be in a flat terrain, e.g., \glspl{uav} flying above a flat water surface.
Based on this, the propagation environment is approximated as a plane reflecting ground surface and therefore modeled by the classical two-ray ground reflection model~\cite[Chap.~4.6]{Rappaport2002}.
This geometrical model for a single receiver is depicted in Fig.~\ref{fig:two-ray-model}.

Based on the setup, it can be seen that the transmitted signal is propagated via two separate paths to the receiver.
On the one hand, there exists a \gls{los} propagation with path length $\len_u$.
On the other hand, the signal is also reflected by the ground, which leads to the second component.
The total length of the second ray is $\lenref_u$.
Finally, these two components superimpose at the receiver.

From basic trigonometric considerations, the path lengths can be calculated as
\begin{align}
	\len_u^2 &= (\htx-\hrx[u])^2 + d_u^2\\
	\lenref_u^2 &= (\htx+\hrx[u])^2 + d_u^2\,.
\end{align}
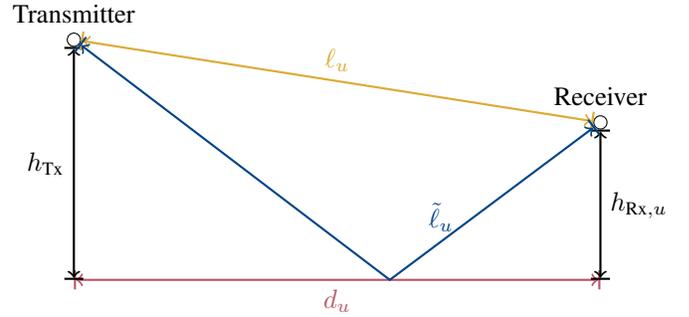
\begin{figure}
	\centering
	\begin{tikzpicture}
	\node[coordinate] (bottomleft) {};
	\node[right=7 of bottomleft,coordinate] (bottomright) {};
	
	\draw[thick,plot2,|<->|] (bottomleft) -- node[below]{$d_u$} (bottomright);
	
	\node[above=3.1 of bottomleft,label={Transmitter},draw,circle,inner sep=0pt,minimum size=5pt] (anttx) {};
	\draw[thick,|<->|] (bottomleft) -- node[left]{$\htx$} (anttx);
	
	\node[above=2 of bottomright,label={Receiver},draw,circle,inner sep=0pt,minimum size=5pt] (antrx) {};
	\draw[thick,|<->|] (bottomright) -- node[right]{$\hrx[u]$} (antrx);
	
	\node[coordinate] at ($(bottomleft)!0.6!(bottomright)$) (reflection) {};
	
	\draw[plot1,thick,|<->|] (anttx) -- node[above]{$\len[u]$} (antrx);
	\draw[plot0,thick,|<->|] (anttx) -- (reflection.center) -- node[above,near start]{$\lenref[u]$} (antrx);
\end{tikzpicture}
	\vspace*{-1.5em}
	\caption{Geometrical model of the considered two-ray ground reflection scenario. The transmitter is placed at height $\htx$ above the ground. The receiver~$u$ is located at height $\hrx[u]$ at a (ground) distance $d_u$ away from the transmitter. The \gls{los} path and reflection path have lengths $\len[u]$ and $\lenref[u]$, respectively.\vspace*{-1.5em}}
	\label{fig:two-ray-model}
\end{figure}

When transmitting on a single frequency $\w=2\pi f$, the received power $\recpower$ of user $u$ at distance $d_u$ is given as~\cite[Eq.~(2)]{Haber1974}, \cite[Chap.~4.6]{Rappaport2002}
\begin{multline}\label{eq:rec-power-single-freq}
	\recpower(d_u, \omega; \htx, \hrx[u], P_t) = \\P_t\left(\frac{c}{2\omega}\right)^2\left(\frac{1}{\len_u^2} + \frac{1}{\lenref_u^2} - \frac{2}{\len_u\lenref_u}\cos\underbrace{\left[\frac{\omega}{c}\left(\lenref_u-\len_u\right)\right]}_{\dphi}\right)
\end{multline}
with the additional notation from Table~\ref{tab:notation}\footnote{In order to simplify the notation, we will omit variables on which functions depend when their value is clear from the context, e.g., we will write $P_r(d)$ instead of $P_r(d, \omega)$ when the value of $\omega$ is fixed.}.
For simplicity, we assume a perfectly reflecting ground throughout this work.
However, all of the calculations can be adapted to include an additional attenuation due to the absorption on the ground.

It is well-known that the two components can interfere constructively or destructively at the receiver, depending on the distance $d$.
This leads to local minima in the receive power at certain distances, which in turn can lead to outages in the transmission.
In order to mitigate these drops in receive power, we propose to use a second frequency in parallel.
However, there exists a set $\frequencies=\{f_1, \dots{}, f_N\}$ of $N$ given frequencies $f_i$ which can be assigned to the individual users.
This directly leads to the question about how the frequencies should be assigned to the users such that the overall reliability of the system is maximized.
The exact problem formulation is described in the following.

\subsection{Problem Formulation}
Throughout the following, we will consider a two-ray ground reflection scenario where the height of the transmitter $\htx$, the height of the receivers $\hrx[u]$ are fixed and known to all parties.
In contrast, the distances between transmitter and receivers $d_u$ may vary and are unknown at the transmitter.
Only the ranges of possible distances are known, i.e., $d_u\in[\dmin[u], \dmax[u]]$.

Based on this knowledge, the transmitter needs to assign distinct frequencies $f_i$ from a set $\frequencies$ of given frequencies to the users.
In order to ensure a high reception quality at any distance, the transmitter may employ up to two frequencies in parallel for each user, such that we can compensate for possible destructive interference of the two rays.
This leads to the following problem.

\begin{prob}\label{prob:opt-prob}
	Since the transmitter only knows the range of $d_u$, i.e., that $d_u\in[\dmin[u], \dmax[u]]$, we try to assign the frequencies such that the average \emph{worst-case} receive power of the system is maximized.
	This optimization problem can be formulated as
	\begin{equation}\label{eq:opt-problem}
		\max_{\substack{\assignments\\\abs{\assignments_{{u}}}\leq 2}}\sum_{u=1}^{K}\min_{d_u\in[\dmin[u], \dmax[u]]} \sum_{f_i\in\assignments_{{u}}}\negthickspace\recpower\left(d_u, f_i; \htx, \hrx[u], \frac{P_t}{\abs{\assignments_{{u}}}}\right),
	\end{equation}
	where $\assignments_{{u}}$ is the set of frequencies that are assigned to user~$u$ and $\assignments=\{\assignments_{1}, \dots{}, \assignments_{K}\}$ is the overall assignment.
	Since each frequency may only be assigned to one user, we additionally have that $\assignments_{{i}}\setintersect\assignments_{{j}}=\emptyset$ for $i\neq j$.
	When using multiple frequencies, the transmit power needs to be split up, which leads to a transmit power $P_t/\abs{\assignments_{{u}}}$ for each individual frequency of user $u$.
\end{prob}
\section{Minimum Receive Power}\label{sec:min-receive-power}
In order to solve \eqref{eq:opt-problem}, we first need to analyze the minimum receive power for a given interval of distances $[\dmin, \dmax]$.

In the following, we start with the case that only a single frequency is used and later extend the results to two frequencies, which are used in parallel.

Throughout this section, we focus on a single user $u$ and therefore omit the index $u$ at the variables.

\subsection{Single Frequency}\label{sub:single-freq}
As mentioned above, it is possible that the two rays interfere destructively which causes a drop in the receive power.
Therefore, we first investigate the distances at which such drops in receive power occur.

\subsubsection{Destructive Interference}
From \eqref{eq:rec-power-single-freq}, it can be seen that we get a (local) minimum of the receive power when the direct and reflected signals interfere destructively~\cite{Loyka2001}.
This occurs whenever the phase difference $\dphi$ is a multiple of $2\pi$, i.e.,
\begin{equation}\label{eq:condition-phase-shift-minimum-single-freq}
	\dphi = \frac{\omega}{c}\left(\lenref-\len\right) = 2\pi k, \quad k\in\naturals_0\,.
\end{equation}
It can easily be verified that $\dphi$ is a decreasing function in $d$ and it therefore follows that
\begin{align*}
	\dphi_{\text{max}}
	&= \lim\limits_{d\to 0} \dphi\\
	&= \frac{2\omega\min\{\htx,\hrx\}}{c} %
\end{align*}
and
\begin{equation*}
	\lim\limits_{d\to\infty} \dphi = 0\,.
\end{equation*}
This shows that $\dphi$ decreases from a finite value $\dphi_{\text{max}}$ to $0$.
Hence, there always exists a finite number of multiples of $2\pi$ with $k_{\text{max}} = \floor{\frac{\dphi_{\text{max}}}{2\pi}}$, i.e., there exist $k_{\text{max}}$ local minima of the receive power.

The distance $d_k$ at which the $k$-th minimum occurs, is given by solving \eqref{eq:condition-phase-shift-minimum-single-freq} as
\begin{equation}\label{eq:distance-k-min-single-freq}
	d_k^2(\omega) = \frac{\left({(c \pi k)^2 - (\omega \hrx)^2}\right) \left({(c \pi k)^2 - (\omega \htx)^2}\right)}{(\omega c \pi k)^2}\,
\end{equation}
with $k=1, \dots{}, k_{\text{max}}$.

\begin{example}\label{ex:single-freq-local-min}
	An illustration of $\dphi$ can be found in Fig.~\ref{fig:phase-shift-single-freq}.
	The parameters are set to $\w/c=10$ ($f=\SI{477}{\mega\hertz}$), $\htx=\SI{10}{\meter}$, and $\hrx=\SI{1.5}{\meter}$.
	Additionally, we indicate the distances $d_k$.
	Since $k_{\text{max}}=4$, there exist four $d_k$ at which a local minimum occurs.
	For the selected parameters, they are evaluated to $d_1=\SI{46.7}{\meter}$, $d_2=\SI{21.6}{\meter}$, $d_3=\SI{12.3}{\meter}$, and $d_4=\SI{6.5}{\meter}$.
	The corresponding received power from \eqref{eq:rec-power-single-freq} is shown in Fig.~\ref{fig:rec-power-single-freq}.
	It can be seen that a minimum occurs at $d=d_k$, with the lowest peak being at the smallest $k$, i.e., $k=1$, which corresponds to the highest distance of all $d_k$.
	For comparison, we additionally show the received power for $f_2=\SI{2.4}{\giga\hertz}$ in Fig.~\ref{fig:rec-power-single-freq}.
	\begin{figure}
		\centering
		\begin{tikzpicture}%
	\begin{axis}[
		width=.95\linewidth,
		height=.25\textheight,
		xlabel={Distance $d$ [\si{\meter}]},
		ylabel={Phase Shift $\Delta\phi$},
		xlabel near ticks,
		ylabel near ticks,
		xmin=1,
		xmax=1000,
		xmode=log,
		ymin=0,
		cycle list name=colorcycle,
		ytick distance=10,
		xtick={1, 10, 100, 1000},
		xticklabels={$10^0$, {}, $10^2$, $10^3$},
		extra x ticks={46.66, 21.64, 12.33, 6.47},
		extra x tick labels={$d_1$, $d_2$, $d_3$, $d_4$},
		extra y ticks={6.28, 12.57, 18.85, 25.13},
		extra y tick labels={$2\pi$, $4\pi$, $6\pi$, $8\pi$},
		extra x tick style={grid=none},
		extra y tick style={grid=none},
		ymajorgrids,
		xmajorgrids,
		xminorgrids,
		grid style={line width=.1pt, draw=gray!20},
		major grid style={line width=.25pt,draw=gray!30},
		]

		\addplot+[domain=1:1000] {10*(sqrt(132.25+x^2)-sqrt(72.25+x^2))};
		\addplot[black,thick,dashed] coordinates {(1, 6.28) (46.66, 6.28) (46.66, 0)};
		\addplot[black,thick,dashed] coordinates {(1, 12.57) (21.64, 12.57) (21.64, 0)};

		\addplot[black,thick,dashed] coordinates {(1, 18.85) (12.33, 18.85) (12.33, 0)};
		\addplot[black,thick,dashed] coordinates {(1, 25.13) (6.47, 25.13) (6.47, 0)};
	\end{axis}
\end{tikzpicture}
		\vspace*{-1em}
		\caption{%
			Relative phase shift $\dphi$ from \eqref{eq:condition-phase-shift-minimum-single-freq} for $\omega/c=10$, $\htx=\SI{10}{\meter}$, and $\hrx=\SI{1.5}{\meter}$.
			Additionally the distances $d_k$, $k=1, \dots{}, 4$, from \eqref{eq:distance-k-min-single-freq} are indicated.
			(Example~\ref{ex:single-freq-local-min})
		}
		\label{fig:phase-shift-single-freq}
	\end{figure}
\end{example}

\subsubsection{Worst-Case Receive Power}
When $d$ is anywhere in the interval $[\dmin, \dmax]$, the lowest peak is at the smallest $k$ such that $d_k$ is still in $[\dmin, \dmax]$.
However, in order to determine the global minimum of the receive power in $[\dmin, \dmax]$, the boundary points $\dmin$ and $\dmax$ need to be taken into account.
This leads to the following result of the minimal receive power when only a single frequency is used.

\begin{thm}[Minimal Receive Power (Single Frequency)]\label{thm:min-rec-power-single-freq}
	Consider the described two-ray ground reflection model with a single frequency $\omega=2\pi f$.
	The distance $d$ between transmitter and receiver is in the interval $[\dmin, \dmax]$.
	The minimal receive power is then given as
	\begin{multline}\label{eq:min-rec-power-single-freq}
		\min_{d\in[\dmin, \dmax]} P_{r}(d) = \\\min \left\{P_r(\dmin), P_r(\dmax), P_r\left(\max_{d_k\in[\dmin, \dmax]} d_k\right)\right\}\,.
	\end{multline}
\end{thm}

\begin{example}[Single Frequency Worst-Case Receive Power]\label{ex:single-freq-worst-case}
	For a numerical example, we take the parameters that are used in Example~\ref{ex:single-freq-local-min} and additionally compare it to a higher frequency scenario.
	In particular, we fix $\htx=\SI{10}{\meter}$, $\hrx=\SI{1.5}{\meter}$, and $P_t=1$.
	The receiver is assumed to be randomly located at a distance between $\dmin=\SI{30}{\meter}$ and $\dmax=\SI{100}{\meter}$ from the transmitter.
	
	For the lower frequency $\omega_1=10c$, i.e., $f_1=\SI{477}{\mega\hertz}$, we get $P_r(\dmin, \omega_1)=\SI{-50}{\decibel}$, $P_r(\dmax, \omega_1)=\SI{-60}{\decibel}$, and $P_r(d_1(\omega_1), \omega_1)=\SI{-97}{\decibel}$ with $d_1(\omega_1)=\SI{46.7}{\meter}$.
	Based on~\eqref{eq:min-rec-power-single-freq} from Theorem~\ref{thm:min-rec-power-single-freq}, we determine that the worst-case receive power is equal to \SI{-97}{\decibel}.
	
	In contrast, for a higher frequency $f_2=\SI{2.4}{\giga\hertz}$, there are multiple local minima at locations $d_k$, which lie in the interval $[\dmin, \dmax]$.
	According to \eqref{eq:min-rec-power-single-freq}, we need to determine the maximum of all $d_k\in[\dmin, \dmax]$.
	For the considered parameters, this is calculated to $d_3(\omega_2)=\SI{79.4}{\meter}$ with $P_r(d_3(\omega_2), \omega_2)=\SI{-125}{\decibel}$.
	The received powers at the boundary points are evaluated to $P_r(\dmin, \omega_2)=\SI{-64}{\decibel}$ and $P_r(\dmax, \omega_2)=\SI{-75}{\decibel}$.
	Hence, the worst-case receive power when using only frequency $f_2$ is $\SI{-125}{\decibel}$.
	\begin{figure}
		\centering
		\begin{tikzpicture}%
	\begin{axis}[
		width=.93\linewidth,
		height=.25\textheight,
		xlabel={Distance $d$ [\si{\meter}]},
		ylabel={Received Power $P_r$ [\si{\decibel}]},
		xlabel near ticks,
		ylabel near ticks,
		xmin=1,
		xmax=1000,
		xmode=log,
		ymax=-40, %
		ymin=-130, %
		legend pos=south west,
		legend cell align=left,
		cycle list name=colorcycle,
		extra x ticks={46.66, 21.64},
		extra x tick labels={$d_1$, $d_2$},
		extra x tick style={grid=none},
		extra y tick style={grid=none},
		ymajorgrids,
		xmajorgrids,
		xminorgrids,
		grid style={line width=.1pt, draw=gray!20},
		major grid style={line width=.25pt,draw=gray!30},
		]
		\addplot+[mark repeat=60] table[x=distance, y=power] {data/power_single-4.771345E+08-t10.0-r1.5.dat};
		\addlegendentry{$f_1=\SI{477}{\mega\hertz}$};
		
		\addplot+[mark repeat=80] table[x=distance, y=power] {data/power_single-2.400000E+09-t10.0-r1.5.dat};
		\addlegendentry{$f_2=\SI{2.4}{\giga\hertz}$};
		
		\addplot[black,dashed] coordinates {(46.66, -40) (46.66, -130)};
		\addplot[black,dashed] coordinates {(21.64, -40) (21.64, -130)};
	\end{axis}
\end{tikzpicture}
		\vspace*{-1em}
		\caption{Received power $P_{r}(d)$ from \eqref{eq:rec-power-single-freq} when using a single frequency $f$ with system parameters $\htx=\SI{10}{\meter}$, $\hrx=\SI{1.5}{\meter}$, and $P_t=1$ for $f=f_1=\SI{477}{\mega\hertz}$ and $f=f_2=\SI{2.4}{\giga\hertz}$. Additionally, the distances $d_1(\omega_1)=\SI{46.7}{\meter}$ and $d_2(\omega_1)=\SI{21.6}{\meter}$ from \eqref{eq:distance-k-min-single-freq} are indicated. (Examples~\ref{ex:single-freq-local-min} and \ref{ex:single-freq-worst-case})}
		\label{fig:rec-power-single-freq}
	\end{figure}
\end{example}

\subsection{Two Frequencies}\label{sub:two-freq}
Since the drops in receive power due to destructive interference of the two rays cannot be avoided when a single frequency is used, we will now employ a second frequency to mitigate these minima.
As described in Problem Statement~\ref{prob:opt-prob}, we aim to optimize the frequency assignment such that the minimum receive power is maximized.

\begin{figure*}[t]%
	\begin{equation}\label{eq:rec-power-two-freq}
		\sumpower = \frac{P_t}{2}\left(\frac{c}{2}\right)^2 \left[\left(\frac{1}{\omega_1^2}+\frac{1}{\omega_2^2}\right)\left(\frac{1}{\len^2}+\frac{1}{\lenref^2}\right)-\frac{2}{\len\lenref}\left(\frac{\cos\left(\frac{\omega_1}{c}(\lenref-\len)\right)}{\omega_1^2} + \frac{\cos\left(\frac{\omega_2}{c}(\lenref-\len)\right)}{\omega_2^2}\right)\right]%
	\end{equation}
	\begin{equation}\label{eq:rec-power-sum-lower-bound}
		\underline{\sumpower}(d, \dw; \omega_1, \htx, \hrx, P_t) = \frac{P_t}{2}\left(\frac{c}{2}\right)^2 \left[\left(\frac{1}{\omega_1^2}+\frac{1}{\omega_2^2}\right)\left(\frac{1}{\len^2}+\frac{1}{\lenref^2}\right)-\frac{2}{\len\lenref}\sqrt{\left(\frac{1}{\omega_1^2}\right)^2 + \left(\frac{1}{\omega_2^2}\right)^2 + \frac{2\cos\left(\frac{\domega}{c}(\lenref-\len)\right)}{\omega_1^2 \omega_2^2}}\right]%
	\end{equation}
	\hrulefill
	\vspace*{-1em}
\end{figure*}

The received power at distance $d$ is given as the sum power $\sumpower = P_{r}(d, \omega_1)+P_{r}(d, \omega_2)$.
For a fair comparison with the single frequency case, we assume that the total transmit power $P_t$ remains the same.
Thus, only half of the transmit power is used for the individual frequencies.
This leads to the expression for the total received power in \eqref{eq:rec-power-two-freq} {at the top of the next page}.

Since we are particularly interested in improving the worst-case performance, we will consider a lower bound on $\sumpower$ in the following.
\begin{lem}[Lower Bound on the Sum Power for Two Frequencies]\label{lem:rec-power-lower-bound-two-freq}
For the described two-ray ground reflection system using two frequencies $\omega_1, \omega_2$ in parallel, the received sum power~$P_s$, is lower bounded by $\underline{\sumpower}$ from \eqref{eq:rec-power-sum-lower-bound} {at the top of the next page}, where $\dw=\w_2-\w_1$ denotes the frequency spacing.
\end{lem}
\begin{proof}
The proof can be found in Appendix~\ref{app:proof-rec-power-lower-bound}.
\end{proof}

Based on this lower bound and Theorem~\ref{thm:min-rec-power-single-freq}, we can immediately state the worst-case receive power in $[\dmin, \dmax]$ when the two frequencies $\w_1$ and $\w_2$ are used in parallel.
\begin{thm}[Minimal Receive Power (Two Frequencies)]\label{thm:min-rec-power-two-freq}
	Consider the described two-ray ground reflection model with two frequency $\omega_1=2\pi f_1$ and $\w_2=2\pi f_2$.
	The distance $d$ between transmitter and receiver is in the interval $[\dmin, \dmax]$.
	The minimal receive power is then given as
	\begin{multline}\label{eq:min-rec-power-two-freq}
		\min_{d\in[\dmin, \dmax]} \underline{\sumpower}(d, \dw) = \min \Bigg\{\underline{\sumpower}(\dmin, \dw),\\ \underline{\sumpower}(\dmax, \dw), \underline{\sumpower}\left(\max_{d_k\in[\dmin, \dmax]} d_k(\dw)\right)\Bigg\}\,.
	\end{multline}
\end{thm}

\begin{example}[Sum Power Lower Bound]\label{ex:sum-power-lower-bound}
As an example, we show the received (sum) power from \eqref{eq:rec-power-two-freq} and the lower bound from \eqref{eq:rec-power-sum-lower-bound} for $f_1=\SI{2.4}{\giga\hertz}$, $\df=\SI{250}{\mega\hertz}$, $\htx=\SI{10}{\meter}$, and $\hrx=\SI{1.5}{\meter}$ in Fig.~\ref{fig:rec-power-two-freq}.
It can clearly be seen that the lower bound is the lower envelope of the actual received power.
Similar to the case that only a single frequency is used, the received power varies with the distance $d$ and shows both minima and maxima.
While the actual received power $\sumpower$ oscillates at a high frequency over the distance $d$, the (spatial) frequency of the lower bound $\underline{\sumpower}$ is determined by the difference in frequencies $\df=\dw/(2\pi)$.
\begin{figure}
	\centering
	\begin{tikzpicture}%
	\begin{axis}[
		width=.93\linewidth,
		height=.25\textheight,
		xlabel={Distance $d$ [\si{\meter}]},
		ylabel={Received Power $P_r$ [\si{\decibel}]},
		xlabel near ticks,
		ylabel near ticks,
		xmin=1,
		xmax=1000,
		xmode=log,
		ymax=-50,
		ymin=-120,
		cycle list name=colorcycle,
		legend cell align=left,
		legend pos=south west,
		ymajorgrids,
		xmajorgrids,
		xminorgrids,
		grid style={line width=.1pt, draw=gray!20},
		major grid style={line width=.25pt,draw=gray!30},
		]
		\addplot+[mark repeat=40] table[x=distance, y=powerSum] {data/power_sum-2.400000E+09-df2.500000E+08-t10.0-r1.5.dat};
		\addlegendentry{Sum Power $\sumpower$};
		\addplot+[mark repeat=40, dashed] table[x=distance, y=envelope] {data/power_sum-2.400000E+09-df2.500000E+08-t10.0-r1.5.dat};
		\addlegendentry{Lower Bound $\underline{\sumpower}$};

	\end{axis}
\end{tikzpicture}
	\vspace*{-1em}
	\caption{Received power for two parallel frequencies with $f_1=\SI{2.4}{\giga\hertz}$, $\df=\SI{250}{\mega\hertz}$, $\htx=\SI{10}{\meter}$, and $\hrx=\SI{1.5}{\meter}$. Both the actual value $\sumpower$ from \eqref{eq:rec-power-two-freq} and the lower bound $\underline{\sumpower}$ from \eqref{eq:rec-power-sum-lower-bound} are shown. (Example~\ref{ex:sum-power-lower-bound})}
	\label{fig:rec-power-two-freq}
\end{figure}

The drops in receive power occur at locations $d_k(\dw)$ which are calculated according to \eqref{eq:distance-k-min-single-freq} to $d_1=\SI{22.9}{\meter}$ and $d_2=\SI{7.5}{\meter}$.
For the distance interval with $\dmin=\SI{30}{\meter}$ and $\dmax=\SI{100}{\meter}$, the minimum receive power is determined by Theorem~\ref{thm:min-rec-power-two-freq} to $\min_{d\in[\dmin, \dmax]}\underline{\sumpower}=\SI{-82.9}{\decibel}$.
Recall from Example~\ref{ex:single-freq-worst-case} that this value is around $\SI{-125}{\decibel}$ when only a single frequency $f=\SI{2.4}{\giga\hertz}$ is used.
This shows that employing a second frequency can significantly improve the worst-case receive power.
\end{example}
\section{Frequency Assignment as Knapsack Problem}\label{sec:freq-assign-qmkp}
Since we are now able to calculate the worst-case receive power for user $u$ with given distance interval $[\dmin[u], \dmax[u]]$ and frequencies $f_1$ and $f_2$, we now want to reformulate the original problem from Problem Statement~\ref{prob:opt-prob} in order to solve it.

The general problem of assigning items to multiple bins (or knapsacks) is referred to as a \gls{mkp}~\cite{Kellerer2004}.
Each item~$i$ has a weight~$w_i$ and each knapsack~$u$ has a weight capacity~$c_u$.
When item~$i$ is placed in any knapsack, it yields the profit $p_i$.
In the \gls{qmkp}, there additionally exist joint profits $p_{ij}$ for placing items~$i$ and $j$ in the same knapsack.

This basic setup closely resembles our considered problem of assigning frequencies to users.
In particular, we have the following correspondences.
The items to be assigned are the $N$ available frequencies $f_i$, $i=1, \dots{}, N$.
The knapsacks, which receive the items, correspond to the individual users $\users = \{1, 2, \dots{}, K\}$.
Since we limit the number of frequencies that can be assigned to each user to two, we set the weight capacity~$c_u$ of each user to $c_u=2$ and the weight of each frequency to $w_i=1$.
The profit $p_i$ of assigning frequency $f_i$ to a user corresponds to the minimum receive power from~\eqref{eq:min-rec-power-single-freq}.
When adding a second frequency $f_j$, the overall profit $p_i+p_j+p_{ij}$ should equal the minimum receive power from~\eqref{eq:min-rec-power-two-freq}. %
However, it needs to be emphasized that these profits are different for each user~$u$, since the parameters $[\dmin[u], \dmax[u]]$ and $\hrx[u]$ may vary between users.
We therefore have the heterogeneous profits
\begin{align}
		p_{u, i} &= \min_{d\in[\dmin[u], \dmax[u]]} P_r(d, f_i) \label{eq:profit-pu-i}\\
		p_{u, ij} &= \min_{d\in[\dmin[u], \dmax[u]]} \underline{\sumpower}(d, f_i, f_j) - p_{u,i} - p_{u, j}\label{eq:profit-pu-ij}\,.
\end{align}
Since this is different compared to the standard definition of the \gls{qmkp}, we generalize it to include heterogeneous profits in the following.

\section{Quadratic Multiple Knapsack Problem with Heterogeneous Profits}\label{sec:qmkp-hp}
For the standard \gls{qmkp} with homogeneous profits, there are multiple (heuristic) algorithms available to find a solution~\cite{Chen2016,Aider2022,Garcia-Martinez2014,Sarac2007,Hiley2006}.
However, since the standard formulation does not include heterogeneous profits, these algorithms do not support it.
While generalizations to the \gls{qmkp} exist, e.g., considering groups of items~\cite{Sarac2014,Chen2016generalized}, they are not directly applicable to our considered problem.

Therefore, we introduce a novel generalization as \gls{qmkp-hp} in the following.

\begin{prob}[Quadratic Multiple Knapsack Problem with Heterogeneous Profits (\Gls{qmkp-hp})]\label{prob:qmkp-hp}
	Let $\frequencies=\{1, 2, \dots{}, N\}$ be a set of items (or objects) and $\users=\{1, 2, \dots{}, K\}$ be a set of knapsacks.
	Each item $i\in\frequencies$ has weight $w_i$ and can be assigned to at most one knapsack $u\in\users$.
	Each knapsack has a weight capacity $c_u$.
	If item $i$ is assigned to knapsack $u$, it yields the profit $p_{u,i}$.
	Additionally, each pair of items $i, j\in\frequencies$, $i\neq j$, has a joint profit $p_{u,ij}$, if both items are assigned to the same knapsack $u$.
	The set $\assignments_u \subset \frequencies$ describes the allocation of items to knapsack $u$.
	
	The objective of the \gls{qmkp-hp} is to find allocations $\assignments=\{\assignments_1, \dots{}, \assignments_K\}$ such that the overall profit is maximized,
	\begin{subequations}\label{eq:qmkp-hp-definition}
		\begin{alignat}{3}
			\max\quad & \sum_{u\in\users}\Bigg(\sum_{i\in\assignments_u} p_{u,i} &+&\sum_{\substack{j\in\assignments_u\\ j\neq i}} p_{u,ij}\Bigg)\\
			\mathrm{s.\,t.}\quad & \sum_{i\in\assignments_u} w_{i} \leq c_u & \quad & \forall u\in\users\\%\text{for all}\; 1\leq u \leq K\\
			& \sum_{u=1}^{K} a_{iu} \leq 1  & & \forall\; 1\leq i \leq N
		\end{alignat}
	\end{subequations}
	where $a_{iu}=1$ if item $i$ is assigned to user $u$, i.e., if $i\in\assignments_u$, and $a_{iu}=0$ otherwise.
\end{prob}

\begin{rem}
	When the profits $P_u$ of the individual knapsacks are all the same, i.e., $P_1=P_2=\cdots{}=P_K$, the \gls{qmkp-hp} reduces to the standard \gls{qmkp}.
\end{rem}

In the following, we develop a greedy algorithm to solve Problem~\eqref{eq:qmkp-hp-definition}.
The algorithm is inspired by the constructive procedure from \cite{Aider2022} for the classical \gls{qmkp}, cf.~\cite[Alg.~1]{Aider2022}.

As an essential guideline in the algorithm, we use the common notion of value density~\cite{Aider2022,Hiley2006}.
This quantity describes the value of adding a particular item~$i$ to a knapsack~$u$ if it already has items $\set{S}$ assigned to it, i.e.,
\begin{equation}\label{eq:def-value-density}
	\vd_{u}(i, \set{S}) = \frac{1}{w_i}\left(p_{u,i} + \sum_{\substack{j\in \set{S}\\i\neq j}} p_{u,ij}\right)\,.
\end{equation}
For a more convenient notation, we define its matrix form as
\begin{equation}\label{eq:def-value-density-matrix}
	\vd(\set{S}) = \begin{pmatrix}
		\vd_1(1, \set{S}) & \vd_2(1, \set{S}) & \cdots{} & \vd_K(1, \set{S})\\
		\vd_1(2, \set{S}) & \vd_2(2, \set{S}) & \cdots{} & \vd_K(2, \set{S})\\
		\vdots{} & \vdots{} & \ddots{} & \vdots{}\\
		\vd_1(N, \set{S}) & \vd_2(N, \set{S}) & \cdots{} & \vd_K(N, \set{S})
	\end{pmatrix}\,.
\end{equation}

After an initialization, the main idea of the proposed Algorithm~\ref{alg:cp-qmkp-hp} at the bottom of the next page is the following.
We start with the combination of item~$i$ and knapsack~$u$ which yields the largest value density.
If there is still room for item~$i$ in knapsack~$u$, i.e., if $w_i \leq c_u$, we add the item to the knapsack.
If there is no room for it, we proceed with the next best combination of item and knapsack, until we can assign an item to a knapsack.
After an item has been assigned, it is removed from the list of unassigned items and the value densities for the remaining ones are updated.
These steps are repeated until there are no more items to assign or all remaining items' weights exceed the remaining capacities of all knapsacks.

\begin{rem}
	Algorithm~\ref{alg:cp-qmkp-hp} can also be used to complete a partial solution $\assignments^{(0)}$.
	We simply need to take $\assignments^{(0)}$ into account while initializing the set of unassigned items.
	If no initial assignments exist, we set $\assignments^{(0)}=\emptyset$.
\end{rem}

\begin{algorithm*}[b]
	\caption{Constructive procedure for solving the \gls{qmkp-hp} from Problem Statement~\ref{prob:qmkp-hp}}
	\label{alg:cp-qmkp-hp}
	\begin{algorithmic}[1]
		\State \textbf{Input:} Items $\frequencies$, Users $\users$, Profits $\profits$, Initial Assignments $\assignments^{(0)}$
		\State \textbf{Output:} Final Assignments $\assignments$
		\Procedure{ConstructiveProcedure}{$\frequencies$, $\users$, $\profits$, $\assignments^{(0)}$}
		\State Initialize unassigned items $\set{S}=\frequencies\setminus \bigcup_u\assignments_{u}^{(0)}$ \Comment{If there is no initial assignment, we have $\set{S}=\frequencies$}
		\State Initialize assignments $\assignments_u = \assignments_{u}^{(0)},\quad\forall u\in\users$ \Comment{If there is no initial assignment, we have $\assignments_{{u}}=\emptyset$}
		\State Calculate initial value densities $V = \vd(\set{S})$ \Comment{According to \eqref{eq:def-value-density-matrix}}
		\Repeat
		\State $\bar{V} = \Call{SortDescending}{V}$ \Comment{Sort $V$ in descending order and store it in $\bar{V}$, i.e., $\max V_{iu}=\bar{V}_1 \geq \bar{V}_2 \geq \dots{}$}
		\State Initialize counter $j=1$
		\Repeat
		\State $\bar{V}_j = V_{iu}$
		\If{$w_i \leq c_u$} \Comment{There is enough space in knapsack $u$}
		\State $\assignments_{{u}} = \assignments_{{u}} \cup \{i\}$ \Comment{Assign element $i$ to user ${u}$}
		\State $\set{S} = \set{S} \setminus \{i\}$ \Comment{Remove element $i$ from set of unassigned items $\set{S}$}
		\State $c_u = c_u-w_i$ \Comment{Adjust the remaining capacity $c_u$ of knapsack $u$}
		\Else
		\State $j=j + 1$ \Comment{Increase counter and proceed with value next in size}
		\EndIf
		\Until{(An item is assigned to a user or $j=\abs{\bar{V}}$)}
		\State $V_{iu} = \vd_{u}(i, \assignments_u),\;\forall u\in\users, i\in\set{S}$ \Comment{Update value densities for the remaining items}
		\Until{($\set{S}=\emptyset$ or $c_u < w_i,\;\forall u\in\users, i\in\set{S}$)}\Comment{All items are assigned or no more space in knapsacks}
		\State \Return $\assignments = \{\assignments_1, \dots{}, \assignments_K\}$
		\EndProcedure
	\end{algorithmic}
\end{algorithm*}

\section{Numerical Examples}\label{sec:numerical-example}
In this section, we demonstrate the effectiveness of our proposed solution to Problem Statement~\ref{prob:opt-prob} by numerical examples.
The source code to reproduce all of the following results can be found at~\cite{BesserGithub}.
All simulations are run in Python on an off-the-shelf laptop computer with Intel~i7-8565 CPU.

We compare different problem sizes with varying numbers of users~$K$ and frequencies~$N$.
For each setup, we uniformly select $N$ frequencies in the interval $\SIrange{2.4}{2.5}{\giga\hertz}$.
The height of the transmitter is fixed to $\htx=\SI{10}{\meter}$.
We then randomly generate $K$ user parameters with $\hrx[u]\sim\unif[1, 3]\si{\meter}$, $\dmin[u]\sim\unif[20, 40]\si{\meter}$, and $(\dmax[u]-\dmin[u])\sim\unif[10, 100]\si{\meter}$.
While we assume $N\geq 2K$ in our examples, this is not necessary.
However, in the case of $N<2K$, not every user gets two frequencies assigned.

Table~\ref{tab:performance-comparison} shows the performance of the proposed \gls{qmkp-hp} approach for different problem sizes.
The reported run times for Algorithm~\ref{alg:cp-qmkp-hp} are averages over \num{100} trials with minimum and maximum in parentheses.
The last two columns indicate the achieved average worst-case receive powers, i.e., the objective of \eqref{eq:opt-problem} divided by the number of users $K$.
For comparison, we additionally evaluated the performance for a random frequency assignment and three \gls{rr} schemes.
The first one (\gls{rr}-simple) starts by assigning user~1 $f_1$, user~2 $f_2$ and continues for each user.
In a second round, each user gets assigned a second frequency, i.e., at the end, user~$u$ is assigned $f_u$ and $f_{u+N}$.
The second \gls{rr} scheme (\gls{rr}-block) is similar to the first one but assigns two frequencies to each user at once, i.e., user~$u$ gets frequencies~$f_{2u-1}$ and $f_{2u}$.
As a final comparison, we add a \gls{rr} scheme that takes the user profits $P_u$ into account (\gls{rr}-profits).
With this, the users take turns selecting the frequency that yields the highest profit for them.

From the results in Table~\ref{tab:performance-comparison}, it can clearly be seen that the proposed algorithm consistently improves the average worst-case receive power by \SIrange{5}{6}{\decibel} over a random frequency assignment.
Additionally, it can be run in feasible times on standard hardware, even for larger problem sizes of \num{100}~frequency blocks and {\num{45}~users}.

In comparison to the \gls{rr} schemes, the proposed \gls{qmkp-hp} approach also shows gains in performance.
The worst performance is found for \gls{rr}-block, since each user is assigned neighboring frequencies, which does not effectively leverage the benefit of using two frequencies in parallel.
The \gls{rr}-simple scheme improves this in general.
However, depending on the number of available frequencies and users, the assigned frequencies can still be very close together.
In contrast, the profit-aware \gls{rr} algorithm achieves results close to the \gls{qmkp-hp} approach.
However, there is still a gap of up to \SI{2.8}{\decibel}, which is maximal in the case where $N$ is close to $2K$.
In these scenarios it can happen in the \gls{rr} scheme that users take frequencies before others, which would have yielded higher benefits for the overall system.

\begin{table*}[hbt]
	\caption{Minimum receive power comparison for the frequency assignment problem (averaged over 100 trials)}
	\label{tab:performance-comparison}
	\vspace*{-1em}
	\centering
	\begin{tabular}{rrcrrrrr}
		\toprule
		$\bm{K}$ & $\bm{N}$ & \textbf{Time Alg.~\ref{alg:cp-qmkp-hp} [$\si{\second}$]} & \textbf{Alg.~\ref{alg:cp-qmkp-hp} [\si{\decibel}]} & \textbf{Random [\si{\decibel}]} & \textbf{\Gls{rr}-Simple [\si{\decibel}]} & \textbf{\Gls{rr}-Block [\si{\decibel}]} & \textbf{\Gls{rr}-Profits [\si{\decibel}]}\\
		\midrule
		\num{3} & \num{10} & \num{0.0093} (\num{0.0089}--\num{0.0117}) & \num{-82.14} & \num{-86.87} & \num{-90.00} & \num{-99.40} & \num{-83.17}\\
		\num{3} & \num{100} & \num{0.8694} (\num{0.7691}--\num{0.9412}) & \num{-81.36} & \num{-87.23} &\num{-108.73} & \num{-113.03} & \num{-81.11}\\
		\num{10} & \num{100} & \num{2.9911} (\num{2.5636}--\num{3.3221}) & \num{-81.17} & \num{-87.21} & \num{-99.44} & \num{-112.60} & \num{-81.61} \\
		\num{20} & \num{50} & \num{1.4986} (\num{1.2925}--\num{1.6559}) & \num{-81.37} & \num{-87.25} & \num{-87.38} & \num{-110.41} & \num{-83.84} \\
		\num{20} & \num{100} & \num{6.1032} (\num{5.1878}--\num{6.4554}) & \num{-81.13} & \num{-87.58} & \num{-93.47} & \num{-112.71} & \num{-82.21}\\
		\num{45} & \num{100} & \num{14.4117} (\num{11.9171}--\num{14.7878}) & \num{-81.69} & \num{-87.56} & \num{-86.62} & \num{-112.70} & \num{-84.51}\\
		\bottomrule
	\end{tabular}
	\vspace*{-1em}
\end{table*}
\section{Conclusion}\label{sec:conclusion}
In this work, we have considered the problem of assigning frequencies to users in a two-ray ground reflection scenario with unknown distances between transmitter and receivers.

We have investigated the receive power in such systems and showed that using two frequencies in parallel can significantly improve the worst-case receive power over using only a single frequency.
In order to solve the optimization problem of assigning up to two frequencies to users for worst-case design, we introduced a generalization of the \gls{qmkp} with heterogeneous profits together with an algorithm to solve it.
Finally, we have illustrated the effectiveness of our approach by numerical examples.

While only two parallel frequencies are considered in this work, it could be extended to multiple frequencies in future work.
This is a promising research direction to further improve the reliability of the described communication systems with limited knowledge at the transmitter.
\appendices
\section{Proof of Lemma~\ref{lem:rec-power-lower-bound-two-freq}}\label{app:proof-rec-power-lower-bound}
The lower bound on $\sumpower$ from \eqref{eq:rec-power-two-freq} is calculated as the (lower) envelope of the function, which is given as the absolute value of its analytic function~\cite{Bracewell2000}.
However, we are only interested in bounding the oscillating part of $\sumpower$ given by the cosine terms as
\begin{equation*}
	s = \frac{\cos\left(\frac{\omega_1}{c}(\lenref-\len)\right)}{\omega_1^2} + \frac{\cos\left(\frac{\omega_2}{c}(\lenref-\len)\right)}{\omega_2^2}\,.
\end{equation*}
The analytic function of $s$ is given as $s+\imag \hat{s}$, where $\hat{s}=\hilbert\{s\}$ is the Hilbert transform $\hilbert$ of $s$.
The envelope of $s$ is then given as $\abs{s+\imag\hat{s}}$.
With the correspondence $\hilbert\{\cos(\omega t)\}=\sin(\omega t)$~\cite{King2009hilbert}, we obtain the analytic signal
\begin{equation*}
	s + \imag\hat{s} = \frac{\cos\left({\omega_1 t}\right)}{\omega_1^2} + \frac{\cos\left({\omega_2}t\right)}{\omega_2^2} + \imag\left(\frac{\sin\left({\omega_1}t\right)}{\omega_1^2} + \frac{\sin\left({\omega_2}t\right)}{\omega_2^2}\right)
\end{equation*}
where we use the shorthand $t=\frac{\lenref-\len}{c}$.
The absolute value can then be calculated as
\begin{equation*}
	\abs{s+\imag\hat{s}}^2 = \frac{1}{\omega_1^4} + \frac{1}{\omega_2^4} + \frac{2 \cos\left((\omega_2 - \omega_1)t\right)}{\omega_1^2 \omega_2^2}\,.
\end{equation*}
Applying the definitions of $t$ and $\dw=\omega_2-\omega_1$ and substituting the envelope of $s$ into \eqref{eq:rec-power-two-freq}, we obtain \eqref{eq:rec-power-sum-lower-bound}.
\printbibliography
\end{document}